\documentclass[a4paper,11pt]{article}
\usepackage{amsfonts,amssymb,amsthm,bm} %
\usepackage[longnamesfirst]{natbib}
\bibpunct{(}{)}{;}{a}{,}{,}
\usepackage[colorlinks,citecolor=blue,urlcolor=blue]{hyperref}
\usepackage[tbtags]{amsmath} %
\usepackage[normalem]{ulem} %
\usepackage{multirow}
\usepackage{enumerate}
\usepackage{fullpage}
\usepackage{fancyvrb}

\usepackage{afterpage}

\usepackage{graphicx,color}

\graphicspath{{./octave/}}
\definecolor{BrickRed}{rgb}{.625,.25,.25}

\definecolor{markergreen}{rgb}{0.6, 1.0, 0}

\usepackage{pdfsync}

\theoremstyle{plain}
\newtheorem{theorem}{Theorem}%[section]
\newtheorem{proposition}[theorem]{Proposition}
\newtheorem{lemma}[theorem]{Lemma} %%
\theoremstyle{definition} %%

\providecommand{\Eref}[1]{Equation~(\ref{#1})} %%
\definecolor{NavyBlue}{rgb}{0,0,.5}

\usepackage{mathrsfs}

\newcommand{\E}{{\mathbb{E}}}

\newcommand{\p}{{\bf P}}
\newcommand{\pas}{\text{{\bf P}--a.s.}}

\newcommand{\e}{{\bf e}}

\newcommand{\dd}{{\rm d}}

\providecommand{\Ncdf}{{\rm N}}
\newcommand{\n}{{\rm n}}

\newcommand{\1}{\textbf{1}}

\ifx\prop\undefined
\newtheorem{prop}{Proposition}[section]
\fi

\newcommand{\argmin}{\operatornamewithlimits{argmin}}

\providecommand{\osigma}{\ensuremath{\overline\sigma}}
\providecommand{\vloan}{\ensuremath{V_{\text{loan}}}}
\providecommand{\vtranche}{\ensuremath{V}}

\begin{document}
\title{\Large\bf Valuation of CDOs on inhomogeneous asset pools with
  an application to structured climate financing} %
\title{\Large\bf Structured climate financing: valuation of CDOs on
  inhomogeneous asset pools} %
\author{N. Packham\footnotemark}%
\maketitle

\begin{abstract}
\noindent Recently, a number of structured funds have emerged as
  public-private partnerships with the intent of promoting  
  investment in renewable energy in emerging markets. These funds
  seek to attract institutional investors by tranching the asset pool
  and issuing senior notes with a high credit quality. Financing of
  renewable energy (RE) projects is achieved via two channels: small
  RE projects are financed indirectly through local banks that draw
  loans from the fund's assets, whereas large RE projects
  are directly financed from the fund. In a bottom-up Gaussian copula
  framework, we examine the diversification properties and RE exposure
  of the senior tranche. To this end, we introduce the LH++ model,
  which combines a homogeneous infinitely granular loan portfolio with
  a finite number of large loans. Using expected tranche percentage
  notional (which takes a similar role as the default probability of a
  loan), tranche prices and tranche sensitivities in RE loans, we
  analyse the risk profile of the senior tranche. We show how the mix
  of indirect and direct RE investments in the asset pool affects the
  sensitivity of the senior tranche to RE investments and how to
  balance a desired sensitivity with a target credit quality and
  target tranche size. 
\end{abstract}

 \medskip%


\section{Introduction}
\label{sec:introduction}

We\renewcommand{\thefootnote}{\fnsymbol{footnote}}%
\footnotetext[1]{%
  \noindent Natalie Packham, Berlin School of Economics and Law,
  Badensche Str.\ 52, 10825 Berlin, Germany. Email:
  packham@hwr-berlin.de \smallskip\\
  \noindent I would like to express my thanks to Jean-David Fermanian,
  Michael Kalkbrener, Ulf Moslener, Radu Tunaru, Ursula Walther and
  Fabian Woebbeking for helpful discussions and comments.} %
\renewcommand{\thefootnote}{\arabic{footnote}}%
consider the problem of valuing and optimally designing structured
finance instruments when the underlying asset pool is
inhomogeneous. The standard in credit portfolio modelling is to assume
a homogeneous credit portfolio, for example in the ``Basel
II''-formula \citep{Gordy2003}, or in the valuation of collateralised
debt obligations (CDOs)
\citep[e.g.][]{Gregory2004,Andersen2004,Hull2007}. In the context of
structured renewable energy financing, asset pools typically consist
of sub-portfolios of different loan types. In this paper, we develop
the necessary tools for pricing and risk management of such structured
products and we explore different aspects of optimally designing asset
pools and related structured products.

To finance sustainable growth in developing economies, governments and
government agencies from various countries, such as Germany, Denmark
and the Netherlands,
% https://www.fmo.nl/climate-fund https://www.ifu.dk/en/services/
are seeking to leverage available financing by attracting private
investors in microfinance investments.\footnote{See e.g.\ the
  following quote from the
  \href{http://www.bmz.de/en/issues/wirtschaft/nachhaltige_wirtschaftsentwicklung/finanzsystementwicklung/mikrofinanzierung/index.html}{website}
  of the Federal Ministry for Economic Cooperation and Development
  (BMZ): ``Involving private investors in microfinance institutions
  and microinsurance funds offers a huge potential. In future, the BMZ
  would like to encourage private-sector involvement to a greater
  extent, and thus facilitate responsible and sustainable investment
  in the financial sector in developing countries. This will also
  spawn numerous opportunities on which other sound development
  projects can build.''} %
Paired with the aim to promote investment in Renewable Energy (RE),
Energy Efficiency, or more generally Green Finance projects, a number
of structured climate funds have been set up as public-private
partnerships.\footnote{See e.g.\ the following quote from the
  \href{http://www.bmz.de/en/issues/wirtschaft/nachhaltige_wirtschaftsentwicklung/finanzsystementwicklung/green_finance/index.html}{website}
  of the BMZ: ``Green finance is an innovative approach of German
  development cooperation. The financial sector in cooperation
  countries becomes part of the transition process to a low carbon,
  resource efficient economy and to improved adjustment to climate
  change.''} %
Examples are the {\em Global Climate Partnership Fund},
\href{http://gcpf.lu}{http://gcpf.lu}, the {\em Green for Growth
  Fund}, \href{http://www.ggf.lu}{http://www.ggf.lu} and the European
Energy Efficiency Fund,
\href{https://www.eeef.eu}{https://www.eeef.eu}. The GCPF, initiated
in 2010 by the German Federal Ministry for the Environment, Nature
Conservation and Nuclear Safety (BMU) and by KfW Development Bank,
issues a junior tranche (C-Shares, first loss, equity tranche), a
mezzanine tranche (B-Shares), a senior tranche (A-Shares) and a
super-senior tranche (Notes). Notes and Class A shares are targeted at
private investors to leverage the amount invested in Class B and C
shares, which are typically held by public investors. In 2018, 91\% of
the fund's asset pool consisted of indirect financing of RE projects
through financial institutions, while 9\% were direct investments in
RE (a significant increase from 1.2\% in Q1/2014).\footnote{GCPF
  Annual Report 2018, GCPF Portfolio Report for the quarter ending on
  31.3.2014}

In such a fund, the equity tranche bears the first losses that occur
in the asset pool. If the equity tranche gets wiped out by defaults,
the mezzanine tranche bears the next losses, etc. This cash flow
structure creates a buffer against credit losses for holders of the
senior tranches,
% In the GCPF, class B and C shares bear the first credit losses in
% the asset pool. The cash flow structure therefore creates a buffer
% against credit losses for Class A share and note holders,
giving the senior tranches a superior credit quality compared to the
lower tranches.\footnote{%
  For an introduction to structured credit, the reader is referred to
  \citep[e.g.][]{Duffie2003,OKane2008}.} %
It is this risk transfer that allows to attract private investors who
are seeking high credit quality investments: Institutional investors
who are interested in investments in innovative asset classes, such as
Green Finance or emerging markets, for example to benefit from
diversification effects \citep[e.g.][]{Krauss2009,Dorfleitner2011},
may be prevented from direct investments due to credit risk
constraints such as credit rating restrictions.
%
% In fact, institutional investors are often subject to investment
% constrains expressed in terms of credit risk constraints, such as
% credit rating restrictions. This renders a direct investment in RE
% projects infeasible.  On the other hand, investments in Green
% Finance as well emerging markets can create desirable
% diversification effects to the existing portfolios of institutional
% investors, see e.g.\ \citep{Krauss2009,Dorfleitner2011}.
Creating tranches of different seniority is therefore a mechanism to
attract public and private capital into climate financing.

\begin{figure}[t]
  % created with ipe
  % cropped with pdfcrop
  \centering \includegraphics[scale=.7]{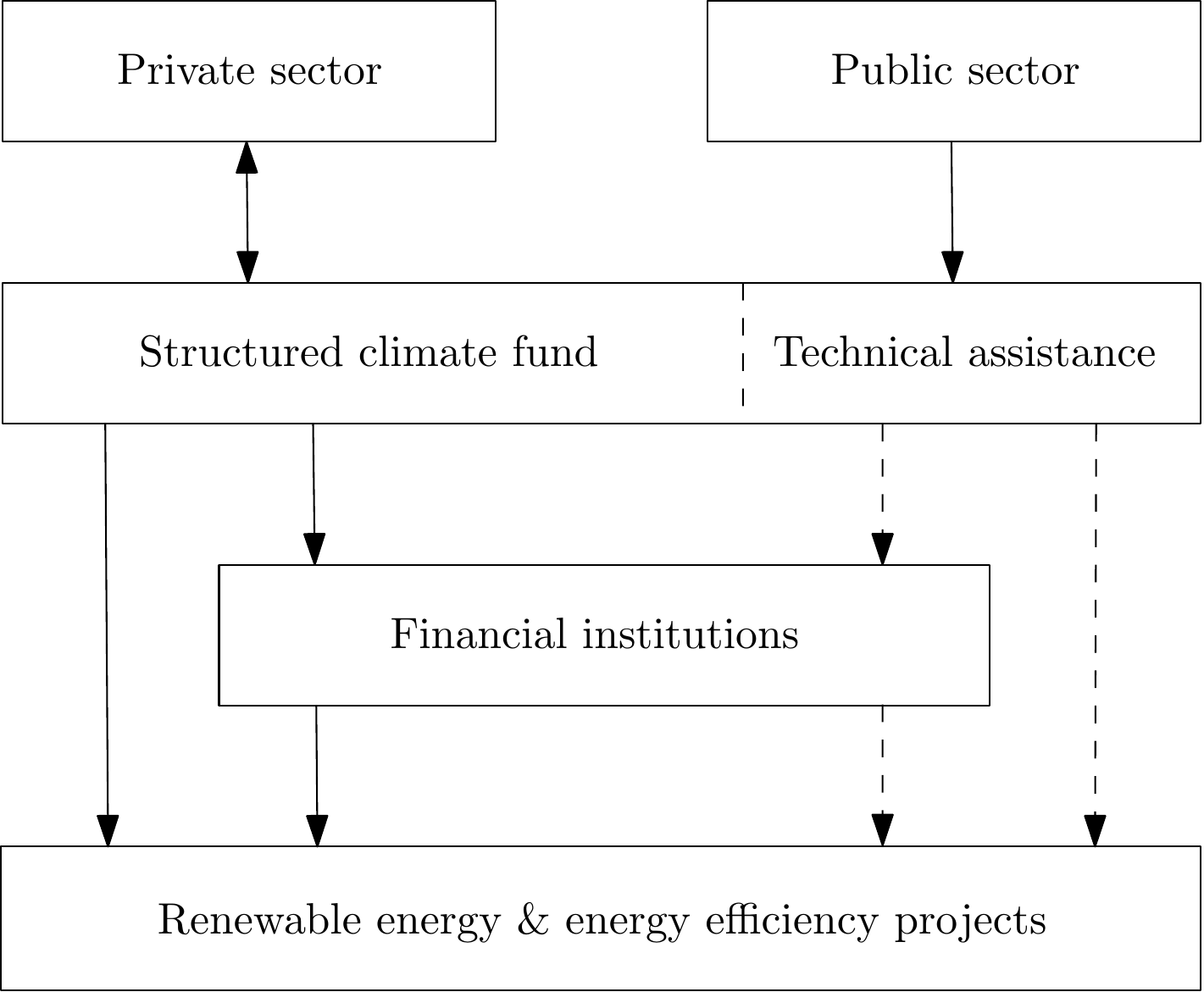}
  \caption{Operational setup of a structured climate fund. Adapted
    from:
    \href{https://www.gcpf.lu/files/assets/images/reporting_news/insights/2018/2018-06-14/GCPF_Corporate_Brochure.pdf}{GCPF,
      Corporate Brochure, 2018}}
  \label{fig:fund_structure}
\end{figure}

As mentioned above, only a small proportion of loans in the asset pool
are direct RE investments.  To enable microfinancing of RE projects,
structured funds typically engage with local banks in developing
countries (see Figure \ref{fig:fund_structure}). These banks act as
intermediaries by drawing funds from the fund's asset pool and lending
them to local borrowers for RE projects. % An additional technical
% assistance (TA) facility may provide operational support.
The asset
pool of such a structured fund consists therefore primarily of claims
against local banks in emerging market and developing countries. While
it is ensured that all capital invested into the structured fund is
channeled into RE projects, the large proportion of indirect financing
creates exposure mainly to regional banks in developing and emerging
countries, and to RE projects only to a lesser extent. This may be
unsatisfactory for an institutional investor seeking exposure to the
RE sector in order to diversify their existing portfolio. On the other
hand, if the asset pool consisted only of (large) direct exposures to
RE projects, diversification could be too low to provide a
reasonably-sized senior tranche.

The objective of this paper is two-fold: First, extending the
LH+-model by \citep{Greenberg2004}, we develop a Merton-type model for
asset pools that consists of a homogeneous infinitely granular
sub-portfolio and a finite number of individual homogeneous loans. We
provide closed pricing formulas for the fund's tranches. Second, we
develop a solution for optimal structuring taking into account both
the exposure towards RE and a desired size of the senior
tranche(s). More specifically, we derive the optimal portfolio mix of
indirect and direct investments maximising the sensitivity to the RE
sector, given a desired size and credit quality of the (super) senior
tranche.

The structure of the paper is as follows: In Section
\ref{sec:model-setup} we develop the LH++ model. Section
\ref{sec:loan-cdo-valuation} provides closed formulas for the
valuation loans and CDO tranches. In Section
\ref{sec:indirect-re-loans}, we determine the parameters (default
probabilities and asset correlations) of banks' when adding a (small)
RE loan to their balance sheet. In Section \ref{sec:cdo-sensitivities}
we introduce CDO tranche sensitivities with respect to the RE sector,
solve for the optimal asset pool structure and senior tranche size and
give examples based on publicly available data on existing structured
funds. Section \ref{sec:conclusion} concludes.

\section{Model setup}
\label{sec:model-setup}

We consider a stylised model, which allows to derive a number of
analytical results. The asset pool consists of loans whose default
probabilities can be determined from a Merton-type asset value
approach \citep{Merton1974}. Two types of loans are found in the asset
pool: a sub-portfolio of homogeneous small loans to banks, which in
turn finance small RE projects; several homogeneous large direct RE
investments. To model these two types of different loans, we extend
the {\em LH+ model\/} developed by \citep{Greenberg2004} to the {\em
  LH++ model\/}, which couples an infinitely granular homogeneous
portfolio with a finite number of large loans. Closed valuation
formulas for CDO tranches are derived in Section
\ref{sec:loan-cdo-valuation}, which in turn allows to analyse the
credit riskiness of senior tranches and their dependence on the loan
structure and correlations in the asset pool.

\subsection{Gaussian copula framework}
\label{sec:gauss-copula-fram-1}

Consider a portfolio of $n$ loans (obligors) and denote their random
times of default by $(\tau_1, \ldots, \tau_n)$. The loss associated
with loan $i$ is $L_i=N_i \cdot (1-R_i)$, where $N_i$ denotes the
exposure-at-default and $R_i$ denotes the recovery rate of loan
$i$. Exposure-at-default and recoveries are assumed to be constant,
that is, neither random nor time-dependent, which is a reasonable
assumption in the context of loans. The overall portfolio loss at time
$t$ is given by
$L_t=\displaystyle\sum_{i=1}^n L_i\, \1_{\{\tau_i\leq t\}}$.

The valuation of loans and CDO tranches requires as input term
structures of default probabilities, both as univariate and as
multivariate distributions. The Gaussian copula framework introduced
by \cite{Li2000} then provides a parsimonious way of modelling these
quantities. Univariate default probabilities, $\p(\tau_i\leq t)$,
$t\geq 0$, $i=1,\ldots, n$, are assumed to be given, for example,
implied from credit spreads observed in the market. If there is no
term structure available in the market, it is common to assume that
the default time follows an exponential distribution with constant
hazard rate $\lambda_i$, so that
\begin{equation}
  \label{eq:24}
  p_{i,t}=\p(\tau_i\leq t) = 1-\e^{-\lambda_i t}, \quad t\geq 0. 
\end{equation}
The hazard rate is calibrated to one given default probability or to
one market credit spread $s_i$ via the relationship
$\lambda_i=s_i/(1-R_i)$.

Joint default probabilities are modelled by the Gaussian copula via
\begin{equation}
  \label{eq:25}
  \p(\tau_i\leq t, \tau_j\leq t) = \Ncdf_2(\Ncdf^{(-1)}(p_{i,t}),
  \Ncdf^{(-1)} (p_{j,t});\rho_{ij}),
\end{equation}
where $\rho_{ij}$ is the so-called {\em asset correlation} and where
$\Ncdf_2(x,y;\rho)$ denotes the bivariate standard normal distribution
function with correlation parameter $\rho$. The terminology asset
correlation originates from the Merton model \citep{Merton1974},
where, for a fixed time horizon $T$, the right-hand side of
(\ref{eq:25}) describes the probability of two firms' ability-to-pay
variables, $Y_{i,T}$ and $Y_{j,T}$ (standard normally distributed)
jointly falling below their so-called default thresholds $c_i,c_j$:
\begin{equation*}
  \p(Y_{i,T}<c_i, Y_{j,T}<c_j) = \Ncdf(c_i,c_j; \rho_{ij}). 
\end{equation*}
One can think of the ability-to-pay variables as a standardised
version of a firm's asset value, and the default thresholds
representing the level of debt, with default occuring if the asset
value falls below the debt level.

In a one-factor model approach, e.g.\ \citet{Vasicek1987}, the asset
correlations enter via a systematic factor, that is, the normalised
asset return can be decomposed as
\begin{equation*}
  Y_{i,t} = \sqrt{\rho_i} V_t + \sqrt{1-\rho_i} \varepsilon_{i,t},
\end{equation*}
with $V_t$ the {\em systematic factor\/} and
$\varepsilon_{1,t},\ldots, \varepsilon_{t,n}$ the {\em idiosyncratic
  factors}, all of which are independent standard normally distributed
random variables.  Conditioning on the joint aggregate factor yields,
for $\rho_i<1$,
\begin{equation}
  \label{eq:12}
  \p(\tau_i\leq t|V_t) = \p(Y_{i,t}\leq \Ncdf^{(-1)}(p_{i,t})|V_t) =
  \Ncdf\left(\frac{\Ncdf^{(-1)}(p_{i,t})-\sqrt{\rho_i} V_t} {\sqrt{1-\rho_i}}\right). 
\end{equation}
Moreover, conditional on $V$, the default times are independent, that
is,
\begin{equation*}
  \p(\tau_1\leq t, \ldots, \tau_n\leq t|V_t) = \prod_{i=1}^n
  \p(\tau_i\leq t|V_t).
\end{equation*}

\subsection{Infinitely granular homogeneous portfolio}
\label{sec:infin-gran-homog}

A common assumption in credit portfolio modelling is to assume that
the portfolio is homogeneous and ``infinitely granular'', which means
that it consists of infinitely many infinitely small homogeneous
components. This stylised portfolio assumption, sometimes called {\em
  large homogeneous portfolio (LHP)\/} was first introduced by
\citep{Vasicek1991} and has many uses in credit portfolio risk; for
example, it provides the basis for the so-called ``Basel II''-formula
\cite{Gordy2003}. In the LHP, idiosyncratic risk is diversified away,
leaving only systematic risk. Formally, by the law of large numbers,
since the obligors are conditionally independent, the percentage loss
conditional on $V_t$ is given by (we drop all indices because of the
homogeneity assumption)
\begin{equation}
  \label{eq:3}
  L_t=\lim_{n\rightarrow\infty} (1-R)\frac{1}{n} \sum_{i=1}^n
  \1_{\{\tau_i\leq t\}} =(1-R)\, 
  \Ncdf\left(\frac{\Ncdf^{(-1)}(p_t)-\sqrt\rho V_t} 
    {\sqrt{1-\rho}}\right) \quad\pas
\end{equation}
Solving for $V_t$, the time-$t$ loss distribution can be written as
\begin{equation*}
  \p(L_t\leq x) =
  \Ncdf\left(\frac{\Ncdf^{(-1)}(x/(1-R))\sqrt{1-\rho}-\Ncdf^{(-1)}(p_t)} 
    {\sqrt{\rho}}\right), \quad x\geq 0. 
\end{equation*}

\subsection{LH++ model}
\label{sec:lh++-model}

As outlined above, a typical asset pool of a structured climate fund
mixes small loans to financial institutions with large direct
financing of RE projects. Hence, the infinitely granular homogeneous
portfolio assumption will be justified only for the part of the
investment pool consisting of loans issued to regional financial
institutions. The {\em LH+ model\/} developed by \citet{Greenberg2004}
incorporates one loan with different characteristics into the asset
pool, which otherwise consists of an infinitely granular homogeneous
portfolio (see also Section 17.3 of \citet{OKane2008}). Extending the
model to allow for several homogeneous loans into the asset pool
conveniently allows to model the direct RE investments. This gives
rise to the {\em LH++ model}, which we introduce here.

Aside from the indirect RE investments through loans to financial
institutions, the asset pool consists of $n$ RE loans that are
modelled each by the Merton asset value model with recovery rate
$R_0$, fractional notional $N_0$, default probability $p_{0,t}$ and
default times $\tau_1,\ldots, \tau_n$. The default probabilities are
linked to standard normally distributed ability-to-pay variables
$X_{1,t}, \ldots, X_{n,t}$ and default threshold $c_0$ via
$p_{0,t} = \p(\tau_k\leq t) = \p(X_{k,t}\leq c_0)$.  The asset
correlation between RE loans is $\rho_0$, which in a model with a
single factor $V_t$ translates into a correlation $\sqrt{\rho_0}$
between an RE loan and the factor, and into a correlation
$\sqrt{\rho\, \rho_0}$ between the homogeneous portfolio and an RE
loan.\footnote{ The model can be extended to include inter-sector
  correlations as well, see e.g.\ \citep{Duellmann2008}. In this
  setting, one would model two sector factors, $V_B$ and $V_{RE}$,
  which are in turn correlated.  }
% The asset correlation between RE loans is $\rho_{RE}$ and the asset
% correlation between the homogeneous portfolio and any one RE loan is
% $\rho_0$.
The fractional notional of the homogeneous portfolio is denoted by
$N$, so that $N+n\, N_0=1$.  The portfolio loss variable conditional
on $V_t$ is then
\begin{equation}
  \label{eq:28}
  L_t = \sum_{k=1}^n N_0(1-R_0)\, \1_{\{\tau_k\leq t\}} + N(1-R)
  \Ncdf\left( \frac{-\sqrt{\rho}V_t + c_t}
    {\sqrt{1-\rho}} \right),
\end{equation}
where $c_t=\Ncdf^{(-1)}(p_t)$ denotes the default threshold associated
with loans in the LHP.

\begin{proposition}
  \label{prop:lossproba}
  The time-$t$ loss probabilities are given by
  \begin{align*}
    \p(L_t> \alpha) = \sum_{k=0}^n \binom{n}{k}
    \p(V_t\leq A_t(\alpha,k), X_{1,t}\leq c_0, \ldots, X_{k,t}\leq c_0,
    X_{k+1,t}>c_0, \ldots, X_{n,t}>c_0),
  \end{align*}
  where
  \begin{equation}
    \label{eq:1}
    A_t(\alpha,k)=\displaystyle\frac{1}{\sqrt{\rho}}
    \left(c_t-\Ncdf^{(-1)}\left(0\vee\left(1 \wedge\frac{\alpha-k
            N_0(1-R_0)} {N(1-R)}\right)\right)\, \sqrt{1-\rho}\right),
  \end{equation}
  with $\vee$ and $\wedge$ denoting the maximum and the minimum
  operator, respectively.  The random vector
  $(V_t, X_{1,t}, X_{2,t}, \ldots, X_{n,t})$ follows a multivariate
  normal distribution with mean vector $0$ and covariance matrix equal
  to the correlation matrix
  \begin{equation*}
    \bm\tilde\Sigma =
    \begin{pmatrix}
      1 & \sqrt{\rho_0} & \sqrt{\rho_0} & \cdots &\sqrt{\rho_0}\\
      \sqrt{\rho_0} & 1 & \rho_0 & \cdots & \rho_0\\
      \sqrt{\rho_0} & \rho_0 & 1 & \cdots & \rho_0\\
      \vdots & \vdots & \ddots & \ddots & \vdots\\
      \sqrt{\rho_0} & \rho_0 & \cdots & \rho_0 & 1
    \end{pmatrix}.
  \end{equation*}
\end{proposition}

\begin{proof}
  Because of the homogeneity of the portfolio of direct RE loans, we
  have
  \begin{align*}
    \p(L_t> \alpha)
    &= \sum_{k=0}^n \binom{n}{k}\p(L_t>\alpha,\tau_1\leq t,\ldots,
      \tau_k\leq t, \tau_{k+1}>t, \ldots, \tau_n>t)\\
    &= \sum_{k=0}^n \binom{n}{k}\p(L_t>\alpha,X_{1,t}\leq c_0,\ldots,
      X_{k,t}\leq c_0, X_{k+1,t}>c_0, \ldots, X_{n,t}>c_0),
  \end{align*}
  and the first claim follows by re-writing $\{L_t>\alpha\}$ using
  \eqref{eq:28}. That $(V_t, X_{1,t}, \ldots, X_{n,t})$ follows a
  joint normal distribution follows from the single factor setting
  with $X_{k,t} = \sqrt{\rho_0} V_t + \sqrt{1-\rho_0} \xi_{k,t}$,
  where $\xi_{k,t}\Ncdf(0,1)$ independent of $V_t$.
\end{proof}

The special case $n=1$ corresponds to the LH+ model, and one can
obtain the formula given by \citep{Greenberg2004}:
\begin{equation}
  \label{eq:2}
  \p(L_t> \alpha) = \Ncdf(A_t(\alpha,0)) 
  - \Ncdf_2\left(A_t(\alpha,0), c_0; \sqrt\rho_0\right)
  + \Ncdf_2\left(A_t(\alpha,1), c_0;\sqrt\rho_0\right),
\end{equation}
where $\Ncdf_2$ denotes the bivariate standard normal distribution
function.

For large $n$, a significant speed-up in the numerical calculation can
be achieved by conditioning on $V_t$ and using the conditional
independence of $X_{1,t}, \ldots, X_{n,t}$, which gives
\begin{align*}
  \p(L_t>\alpha) &=\sum_{k=0}^n \binom{n}{k} \E\left[ \1_{\{V_t\leq
                   A_t(\alpha,k)\}}\, 
                   \Ncdf\left(\frac{c_0-\sqrt{\rho_0} V_t} 
                   {\sqrt{1-\rho_0}}\right)^k
                   \left(1-\Ncdf\left(\frac{c_0-\sqrt{\rho_0} V_t} 
                   {\sqrt{1-\rho_0}}\right)\right)^{n-k}\right] \\
                 &= \sum_{k=0}^n \binom{n}{k}
                   \int_{-\infty}^{A_t(\alpha,k)} \Ncdf\left(\frac{c_0-\sqrt{\rho_0} v}
                   {\sqrt{1-\rho_0}}\right)^k \left(1-\Ncdf\left(\frac{c_0-\sqrt{\rho_0} v}
                   {\sqrt{1-\rho_0}}\right)\right)^{n-k}\, \n(v)\, \dd v,
\end{align*}
where $n$ denotes the standard normal density function.

\section{Loan and CDO valuation}
\label{sec:loan-cdo-valuation}

In this section we derive analytic formulas for valuing loans and CDO
tranches in the LH++ framework.

\subsection{Loan valuation and credit spread}
\label{sec:loan-valuation}

The relation between a loan's credit spread and its survival
probability is as follows: Assume a loan with notional 1, maturing at
$T$ and continuously paying interest of $(r+s)$, where $r$ is the
risk-free interest rate and $s$ is the credit spread. If the loan
defaults prior to maturity, it pays a recovery $R$. The
discounted cash flows from the loan are therefore
\begin{equation*}
  (r+s) \int_0^{T\wedge \tau} \e^{-ru}\, \dd u + \e^{-rT}
  \1_{\{\tau>T\}} + R \e^{-r\tau} \1_{\{\tau\leq T\}}. 
\end{equation*}
At time $0$, the risk-neutral price of the loan is given by
\begin{equation}
  \label{eq:20}
  \vloan(s) = (r+s) % \underbrace{
    \int_0^T \e^{-r u} q(u)\, \dd
    u% }_{=\text{RPV01}}
  + \e^{-r T} q(T) + \int_0^T R\, \e^{-r u}\cdot
  -q(\dd u), 
\end{equation}
where $q(u)=\p(\tau>u)$ is the risk-neutral probability of survival
until time $u$ (conditional on no default until time $0$). % The
% quantity $\text{RPV01}$ denotes the {\em risky present value of a
%   basis point}.
As the no-arbitrage price at inception is $\vloan(s)=1$, the
no-arbitrage spread can be backed out from survival probabilities as
\begin{equation*}
  s_{\text{loan}} = \frac{1-\int_0^T R\, \e^{-ru}\cdot -q(\dd u) -
    \e^{_-r T} q(T)}
  {\int_0^T \e^{-ru} q(u)\, \dd u } - r.
\end{equation*}
In case the term structure of survival probabilities is determined by
a constant hazard rate, $q(u)=\e^{-\lambda u}$, $u>0$, the expressions
simplify to 
\begin{align}
  \vloan(s) &= (r+s) \frac{1-\e^{-(r+\lambda)T}} {r+\lambda} +
              \e^{-(r+\lambda)T} + R \frac{\lambda (1-\e^{-(r+\lambda)T})}
              {r+\lambda} \label{eq:7}\\
  s_{\text{loan}} &=\lambda (1-R)\label{eq:18},
\end{align}
where the last line is the so-called {\em credit triangle}.
% In this case, we shall sometimes find it convenient to write the
% survival probability as a function of the constant spread, that is,
% $q(u,s_{\text{loan}})$.

\subsection{CDO mechanics and valuation}
\label{sec:cdo-valuation}

A CDO can be thought of as a special purpose vehicle consisting of
loans as assets and notes of different seniority as
liabilities. Proceeds from the asset pool, both coupon and redemption
payments, are paid to note holders according to their seniority: on a
payment date, senior note (tranche) holders are the first to receive
their promised coupon and redemption payments. Next, provided there
are sufficient proceeds from the asset pool, mezzanine note (tranche)
holders are served, and so on. Last-in-line are equity tranche
holders. As such, equity tranche holders are exposed to the highest
credit risk, bearing the first losses from the asset pool (the equity
tranche is sometimes called the ``first-loss piece''), while the
senior tranche enjoys a risk buffer as it is unaffected by losses
until the equity and mezzanine tranches are wiped out. For further
details on CDOs, the reader is referred to e.g.\
\citep{Bluhm2003,OKane2008}.

So far we have assumed a fixed default time horizon $T$. Valuation of
credit derivatives requires a term structure of default probabilities.
so that we now assume that all quantities of interest are
time-dependent, e.g.\ default probabilities are given by
$p_{i,t}=\p(\tau_i\leq t)=\Ncdf(c_{i,t})$, with $c_{i,t}$ the
time-dependent default threshold.

We continue to assume that the asset pool consists of an infinitely
granular portfolio of homogeneous obligors for the indirect RE
investments and of $n$ direct investments in homogeneous RE loans. The
tranche structure of a CDO on a notional amount of $1$ can be written
as a partition of $[0,1)$, where each tranche covers the loss in one
interval of the partition. In other words, there exist attachment,
resp.\ detachment points $0<\alpha_1 <\cdots< \alpha_k=1$ such that
the $i$-th tranche covers losses in the interval
$[\alpha_{i-1},\alpha_i)$. The time-$t$ loss of the $i$-th tranche can
be written as
\begin{equation}
  \label{eq:4}
  L_{i,t} =\min(L_t,\alpha_i)-\min(L_t,\alpha_{i-1}) = (L_t\wedge
  \alpha_i) - (L_t\wedge \alpha_{i-1}),
\end{equation}
where the time-$t$ portfolio loss $L_t$ is given by \eqref{eq:28}.
The probability that the $i$-th tranche is hit by a loss until time
$t$, $\p(L_t>\alpha_{i-1})$, is given by Proposition
\ref{prop:lossproba}.

Pricing a CDO tranche with attachment point $\alpha_{i-1}$ and
detachment point $\alpha_i$ requires the time-zero {\em tranche
  survival curve\/}, which expresses the expected percentage survival
notional at time $t$, and is given by
\begin{equation}
  \label{eq:5}
  q_i(t) = q(t,\alpha_{i-1},\alpha_i) =
  1-\frac{\E[L_{i,t}]}{\alpha_i-\alpha_{i-1}} =
  1-\frac{\E\left[(L_t\wedge\alpha_i) -
      (L_t\wedge \alpha_{i-1})\right]}{\alpha_i-\alpha_{i-1}}. 
\end{equation}
An explicit expression for \Eref{eq:5} is obtained from the following
Proposition.
\begin{proposition}
  \label{prop:exploss}
  \begin{align*}
    \E[L_t\wedge \alpha] %
    &= \alpha\, \p(L_t>\alpha)\, + \sum_{k=0}^n \binom{n}{k}
      k N_0(1-R_0)\\ 
    &\phantom{=\, \sum\,\,\,\,} %
      \p(V_t\geq A_t(\alpha,k), X_{1,t}\leq c_0, \ldots, X_{k,t}\leq
      c_0, X_{k+1,t}>c_0, \ldots, X_{n,t}>c_0) \\
    &\phantom{=\,}%
      + \sum_{k=0}^n \binom{n}{k} N(1-R)\\
    &\phantom{=,\sum\,} %
      \p(Y_{1,t}\leq c_t, V_t\geq A_t(\alpha,k), X_{1,t}\leq c_0, \ldots,
      X_{k,t}\leq c_0, X_{k+1,t}>c_0, \ldots, X_{n,t}>c_0),
  \end{align*}
  where $A_t(\alpha,k)$ is given by \eqref{eq:1}.

  The random vector $(Y_{t,1}, V_t, X_{1,t}, \ldots, X_{n,t})$ follows
  a joint normal distribution with mean vector $0$ and covariance
  matrix equal to the correlation matrix
  \begin{equation}
    \label{eq:29}
    \bm \Sigma = 
    \begin{pmatrix}
      1 & \sqrt{\rho} & \sqrt{\rho\, \rho_0} &\sqrt{\rho\, \rho_0} &
      \cdots  & \sqrt{\rho\, \rho_0}\\
      \sqrt{\rho} & 1 & \sqrt{\rho_0} &\sqrt{\rho_0} & \cdots
      & \sqrt{\rho_0} \\
      \sqrt{\rho\, \rho_0} & \sqrt{\rho_0} & 1 & \rho_0 & \cdots& \rho_0\\
      \sqrt{\rho\, \rho_0} & \sqrt{\rho_0} & \rho_0 & 1 & \ddots & \rho_0\\
      \vdots & \vdots & \vdots &  \ddots & \ddots & \vdots \\
      \sqrt{\rho\, \rho_0} & \sqrt{\rho_0} & \rho_0 & \cdots & \rho_0
      & 1
    \end{pmatrix}.
  \end{equation}
\end{proposition}

\begin{proof}
  Write
  \begin{equation*}
    \E[L_t\wedge\alpha] = \E\left[L_t\, \1_{\{L_t\leq \alpha\}}\right] +
    \alpha \p(L_t>\alpha).
  \end{equation*}
  The proof reduces to examining the expectation on the right-hand
  side, which can be written as
  \begin{equation}
    \label{eq:9}
    \E\left[L_t\, \1_{\{L_t\leq \alpha\}}\right] %
    = \sum_{k=0}^n \binom{n}{k} \E\left[L_t\, \1_{\{L_t\leq \alpha,
        \tau_1\leq t, \ldots, \tau_k\leq t, \tau_{k+1}>t\, \ldots,
        \tau_n>t\}}\right]. 
  \end{equation}
  The loss variable $L_t$, given by \eqref{eq:28}, can be decomposed
  into
  \begin{equation}
    L_t = \sum_{k=1}^n N_0(1-R_0)\, \1_{\{\tau_k\leq t\}} + N(1-R)
    \p(Y_{1,t}\leq c_t|V_t). 
  \end{equation}
  It is easily checked that, conditional on $k$ RE loan losses,
  $\{L_t\leq \alpha\} = \{V_t\geq A_t(\alpha,k)\}$.  Using that
  $\{\tau_k\leq t\}=\{X_k\leq c_0\}$, $k=1,\ldots, n$, each
  expectation on the right-hand side in \eqref{eq:9} can be written as
  \begin{multline*}
    \E\left[L_t\, \1_{\{L_t\leq \alpha, \tau_1\leq t, \ldots,
        \tau_k\leq
        t, \tau_{k+1}>t\, \ldots, \tau_n>t\}}\right]\\
    \begin{aligned}
      &= k N_0(1-R_0) \p(A_t(\alpha,k)\leq V_t, X_{1,t}\leq c_0,
      \ldots, X_{k,t}\leq c_0, X_{k+1,t}>c_0, \ldots,
      X_{n,t}>c_0)\\
      &+ N(1-R) \E\left[\p(Y_t\leq c|V_t)\, \1_{\{A_t(\alpha,k)\leq
          V_t, X_{1,t}\leq c_0, \ldots, X_{k,t}\leq c_0,
          X_{k+1,t}>c_0, \ldots, X_{n,t}>c_0\}}\right],
    \end{aligned}
  \end{multline*}
  and the expectation in the last line simplifies to
  \begin{multline*}
    \E\left[\p(Y_t\leq c|V_t)\, \1_{\{A_t(\alpha,k)\leq V_t,
        X_{1,t}\leq c_0, \ldots, X_{k,t}\leq c_0, X_{k+1,t}>c_0,
        \ldots,
        X_{n,t}>c_0\}}\right]\\
    \begin{aligned}
      &= \E\left[\E\left[\p(Y_t\leq c|V_t)\, \1_{\{A_t(\alpha,k)\leq
            V_t, X_{1,t}\leq c_0, \ldots, X_{k,t}\leq c_0,
            X_{k+1,t}>c_0, \ldots,
            X_{n,t}>c_0\}}|V_t\right]\right]\\
      &= \E\left[\p(Y_t\leq c|V_t)\, \p(A_t(\alpha,k)\leq V_t,
        X_{1,t}\leq c_0, \ldots, X_{k,t}\leq c_0, X_{k+1,t}>c_0,
        \ldots,
        X_{n,t}>c_0|V_t)\right]\\
      &= \p(Y_t\leq c,A_t(\alpha,k)\leq V_t, X_{1,t}\leq c_0, \ldots,
      X_{k,t}\leq c_0, X_{k+1,t}>c_0, \ldots, X_{n,t}>c_0),
    \end{aligned}
  \end{multline*}
  where the last line follows from the conditional independence of
  $Y_t$ and the other variables given $V_t$.
\end{proof}

If $n=1$ (the LH+ model), we obtain:
\begin{multline*}
  \E[L_t\wedge \alpha] %
  = \alpha \p(L_t > \alpha) %
  + N_0 (1-R_0) \left[ \Ncdf(c_0) - \Ncdf_2(A_1(\alpha,1), c_0;
    \sqrt{\rho_0})\right]\\ %
  + N(1-R)\left[ \Ncdf(c_t) - \Ncdf_2(c_t, A_t(\alpha,0); \sqrt{\rho})
    + \Ncdf_3(c_t,A_t(\alpha,0), c_0;\Sigma) - \Ncdf_3(c_t,
    A_t(\alpha,1),c_0;\Sigma)\right],
\end{multline*}
where $\Ncdf_k$ denotes the multivariate standard normal distribution
function for a $k$-dimensional vector.  As noted for Proposition
\ref{prop:lossproba}, the numerical computation of the multivariate
probabilities can be efficiently improved by exploiting the
conditional independence of the terms conditional on $V_t$.

To calculate the tranche survival curve \eqref{eq:5} via Proposition
\ref{prop:exploss} requires the time-$t$ PD's and correlations of the
loan portfolio as inputs. If available, the term structure of default
probabilities, $p(t)$, $t\geq 0$, can be derived from market data, for
example CDS spreads.

The $i$-th tranche pays (continuously) a coupon of $r+s_i$ on the
remaining tranche notional, where $r$ denotes the (constant) risk-free
interest rate. At maturity, the tranche pays the remaining
notional.\footnote{%
  If the recovery rate is greater than zero, then the spread earned by
  the collateral pool does not suffice to pay the required coupons to
  all tranche holders: the total notional is reduced by a fraction
  $1-R$ for each defaulted loan, but coupon payments are reduced by
  the entire notional of the loan as no coupon payments are made on
  the recovery rate. There are essentially two ways to resolve this
  discrepancy: In the first case, the notional based on which coupons
  are paid on the super senior tranche is reduced, therefore
  effectively reducing coupon payments on the super senior tranche
  (but without affecting the redemption of notional at maturity), cf.\
  Section 12.5.4 of \citep{OKane2008}. In the second case, coupon
  payments are paid according to the waterfall principle, that is,
  first the promised coupon payment to the super senior tranche is
  made, then to the senior tranche, and so forth, with the remainder
  paid to the equity tranche. This is the case treated in
  \citep{Bluhm2003a}. We shall essentially follow the second
  convention here, as it is natural to assume that the public
  institution (e.g.\ government) as the equity tranche holder is
  willing to waive its coupon anyway. On top, we shall assume that the
  public institution is prepared to ensure that all tranches (with the
  exception of the equity tranche) receive a fixed coupon payment
  proportional to the remaining tranche notional in case the
  collateral pool fails to generate the promised coupons. } By
risk-neutral valuation, the {\em percentage\/} value of the $i$-th
tranche at time $0$ is given by
\begin{equation}
  \label{eq:21}
  \vtranche(\alpha_{i-1}, \alpha_i; s) = (r+s) \int_0^T \e^{-ru} q_i(u)\, \dd
  u
  + \e^{-rT} q_i(T). 
\end{equation}
At inception, the no-arbitrage\footnote{%
  If a CDO tranche can be hedged, for example with a synthetic CDO
  tranche valued 0 at inception or with the reference portfolio, then
  the no-arbitrage price of $1$ arises. If the tranche cannot be
  replicated, then we define the price in this way.} %
percentage value of the tranche is $1$, so that backing out the spread
yields
\begin{equation*}
  s_i = \frac{1-\e^{-rT} q_i(T)}{\int_0^T \e^{-ru} q_i(u)\, \dd u} -r
\end{equation*}
Consequently, given PD's and correlations, the valuation of CDO
tranches can be done in an analytic way.

\section{Indirect renewable energy financing}
\label{sec:indirect-re-loans}

In this section, we derive the model parameters of a bank that draws
on the asset pool to lend out an RE loan. This has an impact on the
bank's credit quality and exposure to RE. We continue to work in a
Gaussian copula framework, but in order to determine the parameters
from enlarging the bank's balance sheet we now make the underlying
Merton model explicit.

The key idea of the Merton model \citep{Merton1974} is to model the
balance sheet of a firm that finances its assets by a single
zero-coupon bond maturing at time $T$ and equity. The firm defaults at
time $T$ if the asset value is below the debt notional and survives
otherwise. If the asset value is modelled as a Geometric Brownian
motion, then the bond value, probability of default and credit spread
can be determined from the Black-Scholes-Merton model.

More specifically, firm $i$ defaults when its time-$T$ asset value
$A_T^i$ is below its debt-value $D_T^i=\e^{rT}D_0^i$, where the
initial debt value $D_0^i$ is constant. If the asset value process
$(A_t^i)_{t\geq 0}$ follows a Geometric Brownian motion,
\begin{equation}
  \label{eq:26}
  A_t^i = A_0^i \, \e^{(r-1/2 \sigma_i^2) t + \sigma_i\, W_t},\quad
  t\geq 0,
\end{equation}
with $W$ a Brownian motion, then the time-$T$ asset log-return is
normally distributed
$\displaystyle \ln(A_T^i/A_0^i) \sim \Ncdf((r-1/2 \sigma_i^2) T,
\sigma_i^2 T)$, and the time-$T$ default probability of obligor $i$,
conditional on $\{\tau_i>0\}$, can be expressed as
\begin{equation}
  \label{eq:8}
  p_i =\p(\tau_i\leq T)= \p(A_T^i<D_T^i ) %
  = \p(Y_i<c_i) = \Ncdf(c_i),
\end{equation}
where $Y_i$ is standard normally distributed and
$c_i=\displaystyle\frac{\ln(\e^{rT} D^i_0/A_0^i) - rT + \sigma_{A^i}^2
  T/2}{\sigma_{A^i}\, \sqrt{T}}$.  Given the time-$T$ probability of
default $p_i$, we define $c_i:= \Ncdf^{(-1)}(p_i)$.

The dependence between two obligors $i$ and $j$ is expressed via their
{\em asset correlation}, given by
\begin{equation*}
  \rho_{ij} = \text{Corr}(Y_i, Y_j) = \text{Corr}(\ln A_T^i, \ln A_T^j),
\end{equation*}
and the probability of a joint default is given by
\begin{equation}
  \label{eq:22}
  \p(\tau_i\leq T, \tau_j\leq T) = \p(Y_i\leq c_i, Y_j\leq c_j) =
  \Ncdf(c_i, c_j; \rho_{ij}) = \Ncdf(\Ncdf^{(-1)}(p_i),
  \Ncdf^{(-1)}(p_j);\rho_{ij}). 
\end{equation}
Equation (\ref{eq:22}) is just the Gaussian copula framework
introduced in \eqref{eq:25}.

We now assume a bank with asset value $A_0$ and debt value $D_0$ at
time $0$.  Adding an RE loan with face value $R_0$ to the balance
sheet changes the asset value to $A_0+R_0$ and the debt value to
$D_0+R_0$. We assume that both the firms debt and the RE loan mature
at time $T$.  Prior to adding the RE loan, the bank's asset volatility
is $\sigma_B$, the bank's time-$T$ probability of default is $p_B$ and
the correlation among any two bank's is $\rho_B := \rho$. The firm
receiving the RE loan has an asset volatility of $\sigma_{R}$, PD of
$p_{R}$, and RE firms are correlated with asset correlation
$\rho_R:=\rho_0$. The bank's asset value (prior to issuance of the RE
loan) and the RE firm's asset value are correlated with
$\rho_{RB}:=\sqrt{\rho\, \rho_0}$.

We impose that after issuance of the RE loan, the assets' log return
is normally distributed,
\begin{equation*}
  \ln\left(\frac{A_T+R_T}{A_0+R_0}\right)\sim
  \Ncdf\left((r-1/2\osigma^2), \osigma^2\right). 
\end{equation*}

Pricing CDO tranches requires the bank's probability of default, which
in turn requires the bank's asset volatility, and the banks' asset
correlations. Assuming that the RE loan is small relative to the
bank's balance sheet, we approximate the {\em annual\/} log-return via
a first-order Taylor expansion around $A_1$,
\begin{equation}
  \label{eq:27}
  \ln(A_1+R_1) \approx \ln(A_1) + \frac{R_1}{A_1}. 
\end{equation}

\begin{proposition}
  \label{prop:3}
  Using the approximation \eqref{eq:27} gives
  \begin{align*}
    \osigma^2
    &= \text{Var}\left(\ln(A_1) + \frac{R_1}{A_1}\right)\\
    &= \sigma_B^2 + \frac{R_0^2}{A_0^2} \e^{2\sigma_B(\sigma_B -
      \rho_{RB} \sigma_R)} (\e^{\sigma_B^2 + \sigma_R^2 - 2 \rho_{RB}
      \sigma_B \sigma_R} -1) %
      - 2\sigma_B\frac{R_0}{A_0} (\sigma_B - \rho_{RB} \sigma_R)
      \e^{\sigma_B^2 - \rho_{RB}\sigma_R\sigma_B}, \\
    \overline{\rho}_{ij} &= \text{Corr}\left(\ln(A_1^i) +
                           \frac{R_1^i}{A_1^i}, \ln(A_1^j) +
                           \frac{R_1^j}{A_1^j}\right) \\
    &= \Big\{\rho_B \sigma_B^2 - 2\sigma_B \frac{R_0}{A_0}
      \e^{\sigma_B^2 -\rho_{RB} \sigma_B \sigma_R} (\rho_B \sigma_B - \rho_{RB}
      \sigma_R) \Big.\\ 
    & \Big.\phantom{=\,} + \frac{R_0^2}{A_0^2}
      \e^{2\sigma_B^2 - 2 \rho_{RB} \sigma_B \sigma_R} 
      \left(\e^{\rho_B \sigma_B^2 + \rho_R \sigma_R^2 -
      2\rho_{RB}\sigma_B\sigma_R} -1\right)\Big\} \,
      \left(\osigma^2\right)^{-1} \\
    \overline{\rho_{RB}}_{B,RE}
    &= \text{Corr}\left(\ln\left(\frac{A_1^i+R_1^i} {A_0^i+R_0^i}\right),
      \ln\left(\frac{R_1}{R_0}\right)\right) \\%
    &\approx \text{Corr}\left(\ln(A_1^i) + \frac{R_1^i}{A_1^i},
      \ln(R_1)\right)\\ %
    &= \left(\rho_{RB} \sigma_B +  \frac{R_0}{A_0}
      \e^{\sigma_B^2 - \rho_{RB} \sigma_B
      \sigma_R} (\rho_R \sigma_R - \rho_{RB}
      \sigma_B) \right) (\osigma)^{-1}. 
  \end{align*}
\end{proposition}
Because the proof consists mainly of long calculations it is deferred
to the appendix. Upon issuance of the RE loan, the bank's PD becomes
\begin{align}
  \p(\overline\tau\leq T)
  &= \p(A_T+R_T\leq \e^{rT}(D_0+R_0))\nonumber\\
  &= \p\left( \ln\left(\frac{A_T+R_T}{A_0+R_0}\right) \leq
    \ln\left(\frac{\e^{rT} (D_0+R_0)}{A_0+R_0}\right)\right) \nonumber\\
  &= \p(\bar Y\leq \bar c),\label{eq:19}
\end{align}
where $\overline Y\sim \Ncdf(0,1)$ and
$\bar c=\displaystyle\frac{\ln(\e^{rT} (D_0+R_0)/(A_0+R_0)) - rT +
  \overline\sigma^2 T/2}{\overline\sigma\, \sqrt{T}}$.

\section{Structuring the asset pool}

Two objectives for structuring the asset pool are important: First,
the asset pool needs to be appropriately diversified, as a
concentrated (i.e., undiversified) asset pool is not capable of
producing sufficient risk transfer between equity and senior
tranches. More specifically, given a target credit quality, the senior
tranche size varies depending on the degree of diversification in the
asset pool. Second, it can be assumed that investors seek exposure to
the RE sector as one of their primary reasons to invest. A typical
institutional investor will therefore find a senior AAA-rated tranche
with a high sensitivity to the RE sector most attractive. Based on
these considerations, we determine the optimal mix of (diversified)
indirect RE loans via banks and direct RE loans in the asset pool.

First, we introduce PV01 and tranche delta to measure the exposure of
an tranche to RE loans. Second, we specify and solve the optimisation
problem to design a structure according to the above-mentioned
criteria.

\subsection{CDO sensitivities}
\label{sec:cdo-sensitivities}

We measure the exposure to RE by the sensitivity of tranche values to
changes in RE loan value changes.  A CDO tranche's {\em PV01 (present
  value of a basis point)} is the change in tranche value following a
one basis point spread widening of the underlying portfolio. The {\em
  (tranche) delta\/} of a CDO tranche is the PV01 relative to the PV01
of the reference portfolio (e.g.\ Chapter 17 \citealp{OKane2008}). The
tranche delta expresses the proportion of the asset pool required to
hedge against changes in the tranche value.

As we are interested in sensitivities with respect to RE, we introduce
the {\em $\text{PV01}_{RE}$} as the value change in a CDO tranche when
the credit spreads of all RE loans (both direct loans and indirect
through bank loans) increase by one basis point. In Section
\ref{sec:cdo-valuation}, we denoted the value of a CDO tranche by
$V(\alpha_{i-1}, \alpha_i, s_i)$. Since we are only considering the
most senior tranche, and need notation for the specific setting, we
denote the tranche value by $V(\lambda, w, \alpha, s)$, where
$\lambda$ denotes the RE loan hazard rate, $w$ denotes the percentage
weight of direct RE loans in the asset pool, $\alpha$ is the senior
tranche's attachment point and $s$ is the credit spread paid on the
tranche. The number of direct RE loans is assumed to be
constant. Using \eqref{eq:18}, a 1 basis point change in the RE loan
spread translates into
$\tilde\lambda = \displaystyle \frac{s_{\text{loan}} + 0.0001}{1-R}$,
giving a sensitivity of
\begin{equation*}
  \text{PV01}_{RE} = V(\tilde\lambda, w, \alpha, s)-V(\lambda,
  w,\alpha, s)
\end{equation*}
and a tranche delta of
\begin{equation*}
  \Delta_{RE} = \frac{\text{PV01}_{RE}}
  {\text{PV01}_{RE,\text{loan}}}, 
\end{equation*}
where $\text{PV01}_{RE,\text{loan}}$ is the PV01 of a single RE loan,
determined from \eqref{eq:7} with $\tilde\lambda$ and $\lambda$,
respectively. 

The value of the direct RE loans in the LH++ model is calculated
directly from \eqref{eq:7}. For the RE loans on the banks' balance
sheets, the new PD is calibrated to the Merton model, \eqref{eq:8}, from
which the new asset volatility of an RE loan is backed out, which in
turn is used to calculate the new PD of the bank portfolio,
\eqref{eq:19}. This is the input to calculating the value of each loan
to a bank, \eqref{eq:20}. For the $\text{PV01}_{RE}$, the previously
calculated quantities enter in the calculation of the tranche survival
curve, Equation (\ref{eq:5}), which in turn enters the tranche
valuation, \eqref{eq:21}.

\subsection{Optimal senior tranche size and RE loan weight}
\label{sec:optim-seni-tranche}

From the structurer's point of view, the objective is to generate a
senior tranche with maximum exposure to RE loans given a desired
minimum size tranche size and a desired credit rating, typically
AAA. The credit rating constraint can be formulated in terms of the
default probability or the expected loss of a AAA-rated loan, cf.\
\citep{Hull2010}. In the first case, one would require
$\p(L\geq \alpha_{i-1}) \leq \pi_{\text{AAA}}$, with
$\pi_{\text{AAA}}$ the PD of AAA-rated loan. In the second case, one
would set $1-q_i(T) \leq \pi_{\text{AAA}} (1-R)$, where $1-q_i(T)$ is
the expected loss as a percentage of the senior tranche's notional,
and $\pi_{\text{AAA}}(1-R)$ is the expected loss of a AAA-rated
loan.\footnote{%
  It should be noted that, since the model is defined under the
  risk-neutral measure, the hitting probability and expected
  percentage loss notional are implied quantities and do not necessary
  coincide with real-world quantities.  } %
In the following, we take expected loss as the constraint.

Let $w\in [0,1]$ be the percentage weight of the direct RE loan
sub-portfolio in the asset pool (i.e., every RE loan has weight
$w/n$). The number $n$ of direct RE loans is assumed to be given --
obviously, at a fixed $w$, a higher $n$ adds diversification, so an
infinitely granular RE sub-portfolio is optimal, but infeasible. Also,
we take the size of RE loans on intermediate banks' balance sheets as
given, as this is a variable that is not controlled by the issuer. The
objective for structuring the asset pool is formulated as the weight
of direct RE loans that maximises exposure to RE, expressed as
$\text{PV01}_{RE}$, while allowing for sufficient diversification in
the asset pool, formulated via a minimum senior tranche size
$\alpha_{\min}$:
\begin{gather}
  \max_{w,\alpha\in[0,1]} |\text{PV01}_{RE}|,\label{eq:31}\\
  \text{subject to \ \ \ \ \ \ \ \ \ \ \ \ \ \ \ \ \ \ \ \ \ \ \ \ \ \
    \
    \ \ }\nonumber\\
  1-q(T,\lambda,w,\alpha) \leq \pi_{\text{AAA}} (1-R),\label{eq:34}\\
  \alpha \leq \alpha_{\max}.\label{eq:35}
\end{gather}
Here, $q(t,\lambda, w,\alpha)$ denotes the expected percentage tranche
notional at time $t\in [0,T]$, cf.\ \eqref{eq:5}.  The constraint
\eqref{eq:35} expresses that the optimal attachment point for $\alpha$
must obey a minimum tranche size, expressed by the maximum attachment
point $\alpha_{\max}$, specified by the issuer. A solution may fail to
exist, if $\alpha_{\max}$ is chosen too small (just consider the case
where $\alpha_{\max}=0$, which is incompatible with the requirement
that the senior tranche attains a AAA rating unless all loans in the
asset pool are AAA-rated).  The following proposition characterises
the solution if it exists.
\begin{proposition}~
  \label{prop:optimal}
  \begin{enumerate}[(i)]
  \item For $w\in (0,1)$, the $\text{PV01}_{RE}$ and the attachment
    point $\alpha\in [0,1]$ satisfy
    \begin{align}
      \text{PV01}_{RE}&<0\label{eq:30}\\[7pt]
      \frac{\partial}{\partial\alpha} \text{PV01}_{RE}&>0. \label{eq:10}
    \end{align}
    As a consequence, if a solution exists (i.e., \eqref{eq:34} and
    \eqref{eq:35} are satisfied), then it~\eqref{eq:34} is binding,
    giving
    $\alpha^\ast(w) = \argmin_{\alpha} q(T,\lambda,w,\alpha) =
    1-\pi_{\text{AAA}} (1-R)$, for $w\in [0,1]$.
  \item For $w\in (0,1)$,
    \begin{equation}
      \label{eq:38}
      \frac{\partial}{\partial w}  \text{PV01}_{RE} <0.
    \end{equation}
    If a solution exists and $w^\ast\in (0,1)$, then
    $\displaystyle \frac{\partial}{\partial w} \alpha^\ast(w^\ast)>0$
    and, as a consequence, \eqref{eq:35} is binding, i.e.,
    $\alpha^\ast(w^\ast) = \alpha_{\max}$. Otherwise, if a solution
    exists, then $w^\ast=1$.
  \end{enumerate}
\end{proposition}
The following Lemma contains some properties that are required for the
proof. 
\begin{lemma}
  \label{lemma:partial}
  Let $V(\lambda, w, \alpha, s)$ denote the value of the most senior
  CDO tranche with attachment point $\alpha$, spread $s$, RE loan
  weight $w$ and RE loan intensity $\lambda$. Then, the following
  properties hold:
  \begin{gather}
    \label{eq:42}
    \frac{\partial}{\partial \lambda} V(\lambda, w, \alpha, s)<0\\
    \label{eq:37}
    \frac{\partial}{\partial w} \frac{\partial}{\partial \lambda} %
    V(\lambda, w, \alpha, s) <0\\
    \label{eq:39}
    \frac{\partial}{\partial s} V(\lambda, w, \alpha, s)>0, \quad
    \frac{\partial}{\partial s} \frac{\partial}{\partial \lambda} %
    V(\lambda, w, \alpha, s) <0\\
    \label{eq:40}
    \frac{\partial}{\partial \alpha} V(\lambda, w, \alpha, s)>0, \quad
    \frac{\partial}{\partial \alpha} \frac{\partial}{\partial
      \lambda} %
    V(\lambda, w, \alpha, s)<0
  \end{gather}
\end{lemma}
\begin{proof}
  The properties all follow from no-arbitrage arguments. \eqref{eq:42}
  and \eqref{eq:37} are a direct consequence of a senior CDO tranche
  being a long position in RE loans. For \eqref{eq:39}, inspection of
  the tranche valuation formula shows that $s$ enters only as a cash
  flow, while $\lambda$ affects the expected percentage tranche
  notional $q$, which is lower for higher $\lambda$, hence eliminating
  some of the positive effect of the spread change $s$. Finally, for
  \eqref{eq:40}, the tranche's credit quality increases with $\alpha$,
  but the impact is smaller when $\lambda$ increases.
\end{proof}

\begin{proof}~ (i) \eqref{eq:30} follows directly from \eqref{eq:42}.
  For \eqref{eq:10}, observe first that
  \begin{equation}
    \label{eq:33}
    \frac{\partial}{\partial \alpha} V(\lambda,w,\alpha, s(\alpha)) =
    0,
  \end{equation}
  with $s(\alpha)$ the fair spread for attachment point
  $\alpha$. Because the credit quality increases with a higher
  attachment point,
  $\displaystyle\frac{\partial}{\partial \alpha} s(\alpha)<0$, it
  follows from \eqref{eq:33} that
  $\displaystyle \frac{\partial}{\partial \alpha} q(u,\lambda,
  w,\alpha)>0$, for all $u\in [0,T]$ (that this holds for all
  $u\in [0,T]$ follows from the monotonicity of $q$ in $u$). Because
  $q(u,\lambda,w,\alpha)$ does not depend on $s$, this holds for
  $\tilde\lambda$ as well:
  $\displaystyle\frac{\partial} {\partial \alpha}
  q(u,\tilde\lambda,w,\alpha)>0$. It follows that
  \begin{equation*}
    \frac{\partial}{\partial \alpha} \text{PV01}_{RE} =
    \frac{\partial} {\partial \alpha} V(\tilde\lambda, w, \alpha,
    s(\alpha)) >0. 
  \end{equation*}
  This proves \eqref{eq:30} and \eqref{eq:10}. It follows jointly from
  \eqref{eq:30} and \eqref{eq:10} that a lower attachment point
  creates the greater exposure (sensitivity). Hence, for given $w$,
  the optimal attachment point is as small as possible. By the rating
  constraint \eqref{eq:34}, a target credit quality requires a minimum
  attachment point, which determines $\alpha^\ast(w)$,
  $w\in [0,1]$.\medskip

  (ii) With the binding constraint \eqref{eq:34}, the optimisation
  problem is re-formulated as
  \begin{gather*}
    \max_w \text{PV01}_{RE} (\lambda, w, \alpha^\ast(w),
    s(w,\lambda^\ast(w))),\\
    \text{such that } \lambda^\ast(w)\leq \alpha_{\max}x.
  \end{gather*}

  Because $V(\lambda, w, \alpha^\ast(w), s(w,\alpha^\ast))=1$, for all
  $w\in [0,1]$, it holds that
  \begin{align}
    \label{eq:36}
    \frac{\partial}{\partial w} V(\lambda, w, \alpha^\ast(w),
    s(w,\alpha^\ast(w)) %
    &= \frac{\partial}{\partial w} V(\lambda, w, \alpha,
      s) %
      + \frac{\partial}{\partial s} V(\lambda, w, \alpha, s)
      \left(\frac{\partial s}{\partial w} + \frac{\partial
      s}{\partial \alpha} \cdot \frac{\partial \alpha}{\partial
      w}\right)\\ %
    &\phantom{=\,} \nonumber%
      + \frac{\partial}{\partial \alpha} V(\lambda, w,
      \alpha, s) \cdot \frac{\partial \alpha}
      {\partial w} = 0,
  \end{align}
  where $s=s(w,\alpha^\ast(w))$ and $\alpha=\alpha^\ast(w)$.  It
  therefore suffices to consider
  \begin{align*}
    \frac{\partial}{\partial w} V(\tilde \lambda, w, \alpha^\ast(w),
    s(w,\alpha^\ast(w)) %
    &= \frac{\partial}{\partial w} V(\tilde\lambda, w, \alpha,
      s) %
      + \frac{\partial}{\partial s} V(\tilde\lambda, w, \alpha, s)
      \left(\frac{\partial s}{\partial w} + \frac{\partial
      s}{\partial \alpha} \cdot \frac{\partial \alpha}{\partial
      w}\right)\\ %
    &\phantom{=\,} \nonumber%
      + \frac{\partial}{\partial \alpha} V(\tilde\lambda, w,
      \alpha, s) \cdot \frac{\partial \alpha}
      {\partial w}. 
  \end{align*}
  Observing that
  $\displaystyle \frac{\partial s}{\partial w}, \frac{\partial
    s}{\partial \alpha}, \frac{\partial \alpha}{\partial w}$ do not
  depend on $\lambda$, the properties \eqref{eq:37}, \eqref{eq:39} and
  \eqref{eq:40} from Lemma \ref{lemma:partial} imply
  $\displaystyle\frac{\partial}{\partial w} V(\tilde\lambda, w,
  \alpha^\ast(w), s(w,\alpha^\ast(w)) <0$, which in turn establishes
  \eqref{eq:38}. \medskip

  For the second part, because \eqref{eq:34} is binding, it follows
  that $q(T,\lambda, w, \alpha^\ast(w))$, is constant for all $w>0$,
  hence by the Implicit Function Theorem
  \begin{equation}
    \label{eq:41}
    \frac{\partial \alpha^\star(w)}{\partial w} =
    -\frac{\frac{\partial}{\partial w} q(T,\lambda, w,\alpha)}
    {\frac{\partial}{\partial \alpha} q(T,\lambda, w, \alpha)},
  \end{equation}
  at $\alpha=\alpha^\ast(w)$.  In part (i) it was established that
  $\displaystyle\frac{\partial}{\partial \alpha}
  q(T,\lambda,w,\alpha)>0$.  It remains to analyse
  $\displaystyle\frac{\partial}{\partial w} q(T,\lambda,w,
  \alpha)$. Increasing the weight $w$ of RE loans in the asset pool
  can increase credit quality e.g.\ by diversification or decrease
  credit quality, e.g.\ by concentration in the asset pool, which in
  turn affects the expected percentage tranche notional $q$. If
  increasing $w$ increases the asset pool credit quality, either by
  diversification or because the RE loan PD is small compared to the
  bank loan PD, then $q$ increases. Vice versa, if increasing $w$
  decreases the asset pool credit quality, either by concentration or
  because the RE loan PD is high compared to the bank loan PD, then
  $q$ decreases. Aside from $q$ being monotone (increasing /
  decreasing) in $w$, the only other possible case is that $q$ is
  concave, i.e., small $w$ diversifies, high $w$ concentrates.

  If $q$ is monotone decreasing or concave,
  $\displaystyle \frac{\partial}{\partial w} \alpha^\ast(w)>0$ by
  \eqref{eq:41}, which implies that a higher attachment point $\alpha$
  creates a higher sensitivity \eqref{eq:38} (in magnitude), and the
  attachment point is constrained by the tranche size requirement
  \eqref{eq:35}, giving an inner solution $w^\ast\in (0,1)$, or by
  $w^\ast=1$. Similarly, If $q$ is monotone increasing in $w$, then
  $\displaystyle\frac{\partial}{\partial w} \alpha^\ast(w)<0$,
  implying $w^\ast=1$ since $\alpha^\ast(1)\leq \alpha_{\max}$.
\end{proof}

\subsection{Example}
\label{sec:example}

For a realistic analysis, the example considered uses publicly
available market data as well as size specifications of the GCPF. All
data are specified in Table \ref{tab:setup}. The $\text{PV01}$ of a
10-year RE loan priced at par is determined to be $-8.7281$ basis
points by calculating a new hazard rate
$\displaystyle\tilde\lambda = \frac{s+0.0001}{1-R}$ from (\ref{eq:18})
and plugging this into (\ref{eq:7}).

% \natp{
% \begin{itemize}
% \item \verb+170227\_Are\_MFIs\_different.pdf+, esp.\ Table 2, has
%   ROA of MFI institutions as well as Banks subsample.
% \item \verb+msci-emerging-markets-financials-index-net.pdf+ MSCI
%   Emerging Markets Financials Index (in USD)
% \end{itemize}
% }

%
\begin{table}[t]
  \centering
  \begin{tabular}{p{6.5cm}p{2cm}p{6.5cm}}
    \hline\hline
    Variable & Value & Comment\\\hline
    CDO maturity & 10 years & \scriptsize{Average duration in GCPF around
                              10 years.} \\  % old value: 10 years\\
    Median bank rating & B+ & \scriptsize{Median of 30 banks in
                              GCPF} \\ % old value: BBB \\
    Average RE loan rating & B & \scriptsize{Median of 9 direct
                                 investments in GCPF} \\ % old value: BB\\
    Bank PD (10 years) & $19.9\%$ & \scriptsize{Extrapolated to 10 years
                                    from 5 year default rate of $B$-rated
                                    financials ($10.5\%$)} \\ % old value: 5.60\%\\
    RE loan PD (10 years) & $24.21\%$ & \scriptsize{Average
                                        cumulative default rate of
                                        $B$-rated global corporates}
    \\ % old value: 17.45\%\\ 
    Weight of RE loan in bank balance sheet
             & 1\% & \scriptsize{Verification of several financial
                     institutions in the GCPF asset pool}\\  % old
                   % value: 1\%\\
                   % Weight of direct RE loans in asset pool &
                   % adjustable\\
    Number of direct RE loans in asset pool & 9 & \scriptsize{from
                                                  GCPF} \\ % old
                                                % value: 5\\
    Percentage notional of all direct RE loans
             & 10.61\% & \scriptsize{total percentage notional of
                         direct investments in GCPF}\\
    Senior tranche hitting probability
    (AAA, 10 years) & 0.70\% & \scriptsize{Average cumulative default
                               rate of $AAA$-rated corporates, see
                               Table 24 of \citet{SP2019}} \\
                             % 0.82\%\\
    Asset correlation among RE loans
             & 0.1170 &\scriptsize{Median of historical correlation among 1-day
                        returns of nine largest constituents in  MSCI Global
                        Alternative Energy Index (USD), Dec 2018--Dec
                        2019, Source: finance.yahoo.com}\\ % old value: 0.4\\
    Asset correlation among bank loans
             & 0.1758 & \scriptsize{Median of historical correlation among 1-day
                        returns of nine largest constituents in  MSCI Emerging
                        Markets Financial Index (USD), Dec 2018--Dec
                        2019, Source: finance.yahoo.com} \\  % old value & 0.2\\
    Asset correlation among bank and RE loans
             & 0.1434 & \scriptsize{Calculated in a one-factor model
                        as $\sqrt{0.1170\cdot 0.1758}$}\\
                      % \scriptsize{Correlation coefficient of daily
                      % log-returns of MSCI Emerging
                      % Markets Financial Index (USD) and MSCI Global
                      % Alternative Energy Index (USD), Dec 2018--Dec
                      % 2019, Source: onvista.de} \\% old value: 0\\
    Recovery rate & 25\% & \scriptsize{Corporate Asset recovery rates
                           for senior secured bonds in a AAA CDO
                           tranche are estimated as 17\% (Group 4
                           countries) and 32\% (Group 3 countries),
                           see Table 10 of \citet{SP2015}, where Group
                           3/4 countries are emerging market countries}\\  % old value: 0.2\\
    Equity / asset ratio (banks' balance sheets)
             & 10\% & \scriptsize{Approximation of several financial
                      institutions in the GCPF asset pool} \\
                    % old value: 10\%\\
    \hline\hline
  \end{tabular}
  \caption{Parameters in example. Data sources: \citet{SP2019},
    \citet{SP2015}, \citet{GCPF2019}, % onvista.de,
    finance.yahoo.com.}
  % PD's are taken from \citet{SP2019},
  % GCPF data is taken from \citet{GCPF2019}.}
  % given the respective rating assumptions. \natp{\em [Risk-free rate
  % should be chosen according to a AAA rating and constant hazard
  % rate setup.]}}
  \label{tab:setup}
\end{table}

Figure \ref{fig:basecase1} shows the tranche attachment point
$\alpha^\ast$, tranche spread $s$, tranche delta $\Delta_{RE}$ and
tranche sensitivity $\text{PV01}_{RE}$, when (i) varying the number of
RE loans, keeping the tranche weight $w=10.61\%$ fixed and (ii)
varying the number of RE loans, keeping each loan's weight fixed. In
case (i), the number of loans plays virtually no role, except for the
attachment point, which decreases as an increasing number of loans
improves diversification in the asset pool. The credit spread is
constant, reflecting that variations in the cash flow structure
compared to the AAA-loan used in obtaining the optimal attachment
point can be neglected.

Figure \ref{fig:basecase2} shows the same properties when varying the
weight $w$, while keeping number of RE loans fixed at $n=5$. Each
graph shows three scenarios: (i) the base scenario with the observed
10-year PD of RE loans of $24.21\%$, (ii) a scenario with high credit
quality RE loans (10-year PD $1\%$) and (iii) a scenario with low
credit quality RE loans (10-year PD $40\%$). For comparison purposes,
the bank loan PD is $19.9\%$. The scenarios yield different shapes of
$\displaystyle \frac{\partial}{\partial w} \alpha^\ast(w)$, cf.\ part
(ii) of Proposition \ref{prop:optimal}. Depending on the choice of
$\alpha_{\max}$, which determines the minimum required senior tranche
size, the optimal the optimal RE loan weight $w^\ast$ will be in
$(0,1)$ in scenarios (i) and (ii), whereas in scenario (iii), we
always have $w^\ast=1$ if a solution exists, as the sensitivity
$\text{PV01}_{RE}$ (bottom right) increases with decreasing
$\alpha^\ast(w)$. The base scenario with the data from Table
\ref{tab:setup} is optimal if $\alpha^\ast(0.1061)=\alpha_{\max}$,
which translates into $\alpha_{\max}=0.3168$ for $n=9$ and
$\alpha_{\max}=0.3179$ for $n=5$.

An interesting observation from Figure \ref{fig:basecase2} is that a
low RE PD leads to a higher RE loan sensitivity. Two effects
contribute to this: first, a 1 bp change in the RE loan spread
translates differently into the credit quality change depending on the
initial RE PD level; second, a low RE PD implies a low attachment
point $\alpha^\ast$, which in turn increases the tranche's sensitivity
to RE loans. The latter effect also implies that an increase in
correlations between RE loans and bank loans may fail to 
increase RE loan sensitivity, as the smaller diversification decreases
the senior tranche size. For example, for $n=5$, $\alpha^\ast$ shifts
from $0.3179$ in the base scenario to $0.6707$ when correlations
between bank loans and between RE loans are set to $\sqrt{0.5}$. In
turn, the
$\text{PV01}_{RE}$ shifts from $-0.1$ basis points to  $-0.034$ basis points. 

% Overall, we find that the RE loan weight is the most important driver
% of tranche sensitivity towards RE. Increasing the RE loan weight will
% typically decrease the senior tranche size, but this effect may be
% less pronounced than anticipated.

\begin{figure}[t]
  \centering \includegraphics[scale=1.5]{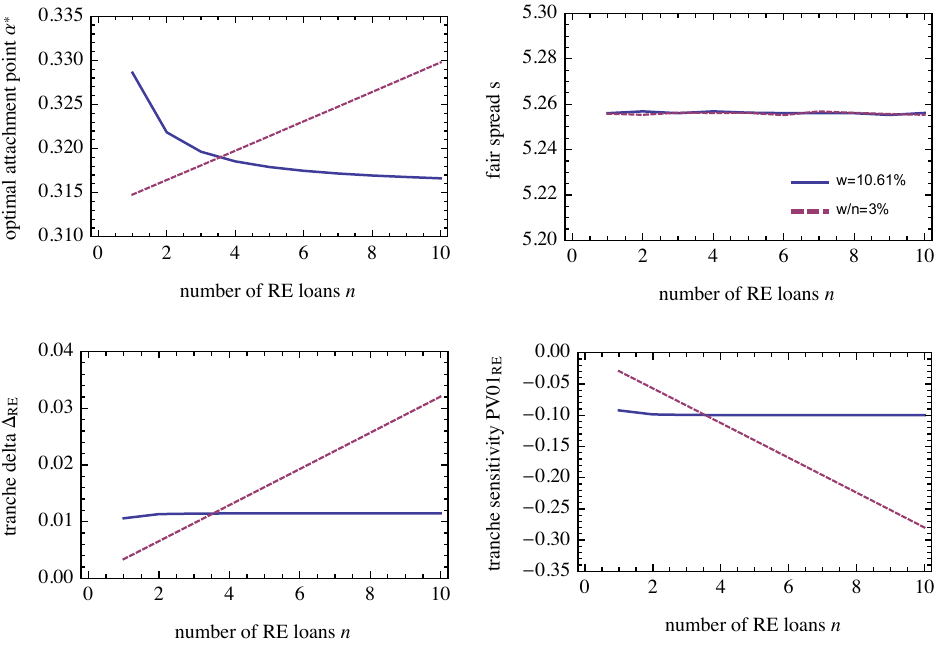}
  \caption{Tranche properties as a function of the number of
    loans. Solid: RE loan weight $w=10.61\%$ is fixed ( each loan has
    weight $w/n$); dashed: each loan has constant weight of $w/n=3\%$.
    Top left: optimal attachment point $\alpha^\ast$; top right: fair
    spread $s$; bottom left: tranche delta $\Delta_{RE}$; bottom
    right: tranche sensitivity $\text{PV01}_{RE}$. }
  \label{fig:basecase1}
\end{figure}

\begin{figure}[t]
  \centering \includegraphics[scale=1.5]{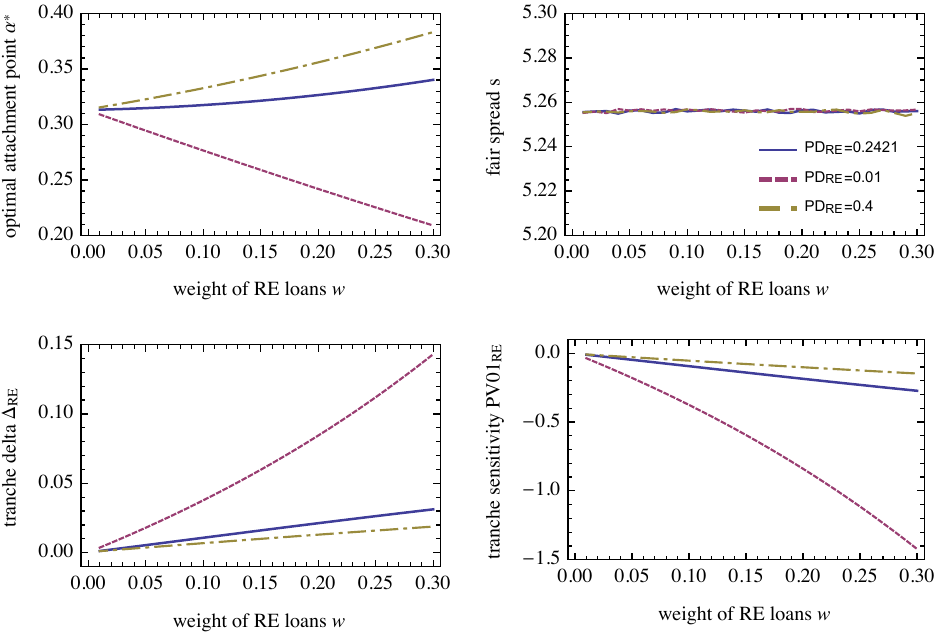}
  \caption{Tranche properties as a function of RE loan weight; number
    of loans is fixed at $n=5$; different PD's for RE loans are
    assumed (see legend in top right graph). Top left: optimal
    attachment point $\alpha^\ast$; top right: fair spread $s$; bottom
    left: tranche delta $\Delta_{RE}$; bottom right: tranche
    sensitivity $\text{PV01}_{RE}$.}
  \label{fig:basecase2}
\end{figure}

\section{Conclusion and Outlook}
\label{sec:conclusion}

We study public-private partnerships that have a CDO-like investment
structure. Here, the public sector invests in the equity tranche,
while institutional investors would typically invest in the senior
tranches. The risk transfer from restructuring the asset pool's cash
flows makes the investment attractive or accessible for risk-averse
institutional investors. These types of financial vehicles have been
issued in a development finance context, with an explicit goal to
promote financing of RE projects in emerging and developing
countries. The asset pool is primarily composed of loans to regional
banks, which in turn provide direct financing of RE
investments. Typically, only few direct investments in RE projects are
contained in the asset pool. As such, although the investment into the
structured fund is channelled into RE projects, this asset pool
composition creates a sensitivity mainly to banks in developing and
emerging markets, and to the RE sector only to a lesser
extent. Assuming that investors would also seek exposure to RE, the
paper provides an answer to questions revolving around the optimal
asset pool setup and structuring.

First, we develop a framework for studying this type of problem by
introducing the LH++ model, a Merton-type model in which the asset
pool consists of an infinitely granular homogeneous portfolio of bank
exposures and one or several large direct RE investment. We derive
closed formulas for CDO tranche valuation, which in turn allow to
calculate sensitivities such as tranche deltas and PV01's against RE
loans. Increasing the proportion of the direct investment increases
the RE exposure. However, since direct investments are larger, this
also decreases diversification in the asset pool, which potentially
decreases the size of the senior tranche (which is characterised by a
target AAA rating).

In our stylised framework, we determine the optimal asset pool mix,
which maximises RE exposure given a minimum senior tranche size and a
desired rating. We show that, in a typical setting, where RE loans
have a lower credit quality than bank loans, the optimal proportion of
RE loans has weight smaller than $1$.

\appendix%
\normalsize

\section{Proof of Proposition \ref{prop:3}}
\label{sec:proofs}

\begin{lemma}
  \label{lemma1}
  Let $X$ and $Y$ be independent, standard normally distributed random  
  variables. Then,
  \begin{align*}
    \E\left[X\, \e^{\sigma_X X}\right] &= \sigma_X\, \e^{\sigma_X^2/2}\\[5pt]
    \E\left[X\, \e^{\sigma_X X + \sigma_Y Y}\right] %
    &= \sigma_X\, \e^{\sigma_X^2/2 + \sigma_Y^2/2}. 
  \end{align*}
\end{lemma}
\begin{proof}
  The first identity is easily calculated to be 
  \begin{equation*}
    \E\left[X\, \e^{\sigma X}\right] = \frac{1}{\sqrt{2\pi}} \int x\,
    \e^{\sigma_X x}  \e^{-x^2/2}\, \dd x
    = \frac{1}{\sqrt{2\pi}} \int x\, \e^{\sigma_X^2/2}\,
    \e^{-\frac{(x-\sigma_X)^2}{2}}\, \dd x
    = \sigma_X\, \e^{\sigma_X^2/2}.
  \end{equation*}
  Using this, the second identity is calculated as
  \begin{align*}
    \E\left[X\, \e^{\sigma_X X + \sigma_Y Y}\right] %
    &= \frac{1}{\sqrt{2\pi}} \frac{1}{\sqrt{2\pi}} \int \int x\, 
      \e^{\sigma_X X + \sigma_Y y} \e^{-x^2/2} \e^{-y^2/2}\, \dd x\,
      \dd y \\ %
    &= \frac{1}{\sqrt{2\pi}} \frac{1}{\sqrt{2\pi}} \int \e^{\sigma_Y
      y}\, \e^{-y^2/2}\, \dd y\, \int x\, \e^{\sigma_X x}\,
      \e^{-x^2/2}\, \dd x \\ %
    &= \e^{\sigma_Y^2/2}\, \sigma_X \, \e^{\sigma_X^2/2}. 
  \end{align*}
\end{proof}

\begin{proof}[Proof of Proposition \ref{prop:3}]
  Let $W^B$ and $W^R$ be the Brownian motions driving the respective
  asset processes. In the calculations below, write $W^R = \rho_{RB}
  W_1^B + \sqrt{1-\rho_{RB}^2} Z$, with $Z\sim \Ncdf(0,1)$ independent
  of $W_1^B$. Recall that the variance of a log-normally distributed
  random variable $\e^X$ with $X\sim \Ncdf(0,\sigma^2)$ is
  $\e^{\sigma^2} (\e^{\sigma^2}-1)$. 
  The variance of the right-hand side of
  \eqref{eq:27} is given by
  \begin{multline*}
    \text{Var}\left(\ln(A_1) + \frac{R_1}{A_1}\right) %
    = \text{Var} \left((r-\frac{1}{2}\sigma_B^2) + \sigma_B W_1^B +
      \frac{R_0}{A_0} \e^{(r-1/2 \sigma_R^2) + \sigma_R W_1^R -
        (r-1/2\sigma_B^2) - \sigma_B W_1^B}\right)\\
    \begin{aligned}
      &= \sigma_B^2 + \frac{R_0^2}{A_0^2}
      \e^{(\sigma_B^2-\sigma_R^2)}\, \text{Var}\left(\e^{(\rho_{RB}\sigma_R
          -\sigma_B) W_1^B + \sqrt{1-\rho_{RB}^2}\sigma_R Z}\right)\\  
      &\phantom{=\,} %
      +2\, \sigma_B\, \frac{R_0}{A_0}\, \e^{(\sigma_B^2-\sigma_R^2)/2}
      \, \text{Cov}\left(W_1^B, \e^{(\rho_{RB}\sigma_R - \sigma_B) W_1^B +
          \sqrt{1-\rho_{RB}^2} \sigma_R Z}
      \right)\\
      &= \sigma_B^2 + \frac{R_0^2}{A_0^2} \e^{(\sigma_B^2-\sigma_R^2)}
      \e^{(\rho_{RB}\sigma_R-\sigma_B)^2 +  (1-\rho_{RB}^2)\sigma_R^2}
      (\e^{(\rho_{RB}\sigma_R-\sigma_B)^2 + (1-\rho_{RB}^2)\sigma_R^2}-1)\\ % 
      &\phantom{=\,} %
     + 2 \sigma_B \frac{R_0}{A_0} \e^{(\sigma_B^2-\sigma_R^2)/2}
     (\rho_{RB}\sigma_R - \sigma_B) \e^{((\rho_{RB}\sigma_R -
       \sigma_B)^2 + \sigma_R^2 (1- \rho_{RB}^2))/2} \\
     &= \sigma_B^2 + \frac{R_0^2}{A_0^2} \e^{2\sigma_B(\sigma_B -
       \rho_{RB} \sigma_R)} (\e^{\sigma_B^2 + \sigma_R^2 - 2 \rho_{RB}
       \sigma_B \sigma_R} -1) %
     + 2\sigma_B\frac{R_0}{A_0} (\rho_{RB} \sigma_R - \sigma_B)
     \e^{\sigma_B^2 - \rho_{RB}\sigma_R\sigma_B}. 
    \end{aligned}
  \end{multline*}

  For the correlation we calculate the required covariance. In the
  calculations we decompose the (correlated) vector  $\mathbf W^T =
  (W_1^{B,i}, W_1^{B,j},  W_1^{R,i}, W_1^{R,j})^T$ into a Cholesky
  factorisation $\mathbf C$ with independent normals $\mathbf Z^T = (Z_1,
  \ldots, Z_4)^T$:
  \begin{equation*}
    \mathbf W = \mathbf C\cdot \mathbf Z,
  \end{equation*}
  where 
  \begin{equation*}
    C =
    \begin{pmatrix}
      1 & 0 & 0 & \\
      \rho_B & \sqrt{1-\rho_B^2} & 0 & 0\\
      \rho_{RB}& \rho_{RB} \frac{1-\rho_B}{\sqrt{1-\rho_B^2}} & \sqrt{\frac{1 +
          \rho_B -2 \rho_{RB}^2}{1+\rho_B}} & 0\\
      \rho_{RB} &  \rho_{RB} \frac{1-\rho_B}{\sqrt{1-\rho_B^2}} & \frac{\rho_R +
        \rho_R \rho_B - 2\rho_{RB}^2} {\sqrt{(1+\rho_B) (1+ \rho_B - 
          2\rho_{RB}^2)}} & \sqrt{\frac{(1+\rho_B) (1-\rho_R^2 ) - 4\rho_{RB}^2
        (1+\rho_R)} {1 + \rho_B - 2\rho_{RB}^2}}
    \end{pmatrix}.
  \end{equation*}
  The covariance is given by: 
  \begin{multline*}
    \text{Cov}\left(\ln(A_1^i) + \frac{R_1^i}{A_1^i}, \ln(A_1^j) +
      \frac{R_1^j}{A_1^j}\right) \\
    \begin{aligned}
      &= \text{Cov}\left(\sigma_B W_1^{B,i}, \sigma_B
        W_1^{B,j}\right) %
      +2 \text{Cov} \left(\sigma_B W_1^{B,i}, \frac{R_0}{A_0}
        \e^{(\sigma_B^2-\sigma_R^2)/2 + \sigma_R W_1^{j,R} - \sigma_B
          W_1^{j,B}}\right)\\ %
      &\phantom{=\,}+ \text{Cov} \left(\frac{R_0}{A_0}
        \e^{(\sigma_B^2-\sigma_R^2)/2 + \sigma_R W_1^{i,R} - \sigma_B
          W_1^{i,B}}, \frac{R_0}{A_0} \e^{(\sigma_B^2-\sigma_R^2)/2 +
          \sigma_R W_1^{j,R} - \sigma_B
          W_1^{j,B}}\right)\\
      &= \rho_B\sigma_B^2 %
      +2 \sigma_B \frac{R_0}{A_0} \e^{(\sigma_B^2 - \sigma_R^2)/2}
      \text{Cov} \left(Z_1, \e^{\sigma_R (c_{41} Z_1 + c_{42} 
          Z_2 + c_{43} Z_3 + c_{44} Z_4)- \sigma_B (c_{21} Z_1 +
          c_{22} Z_2)}\right)\\ %
      &\phantom{=\,}+ \frac{R_0^2}{A_0^2} \e^{(\sigma_B^2-\sigma_R^2)}
      \text{Cov} \left( \e^{\sigma_R (c_{31} Z_1 + c_{32} Z_2 + c_{33}
          Z_3) - \sigma_B Z_1}, \e^{\sigma_R (c_{41} Z_1 + c_{42} Z_2
          + c_{43} Z_3 + c_{44} Z_4) - \sigma_B (c_{21} Z_1 + c_{22}
          Z_2)}\right)\\
      &= \rho_B \sigma_B^2 + 2\sigma_B \frac{R_0}{A_0} \e^{(\sigma_B^2
        - \sigma_R^2)/2} (\sigma_R c_{41} -\sigma_B c_{21})
      \e^{((\sigma_R c_{41}-\sigma_B c_{21})^2 + (\sigma_R c_{42}
        -\sigma_B c_{22})^2 + \sigma_R^2 c_{43}^2 + \sigma_R^2
        c_{44}^2)/2}\\ %
      &\phantom{=\,}+ \frac{R_0^2}{A_0^2} \e^{(\sigma_B^2-\sigma_R^2)}
        \e^{((\sigma_R c_{31} -\sigma_B + \sigma_R c_{41} - \sigma_B
          c_{21})^2 + (\sigma_R c_{32} + \sigma_R c_{42} -
          \sigma_B c_{22})^2 + (\sigma_R c_{33} + \sigma_R
          c_{43})^2 + \sigma_R^2 c_{44}^2)/2}\\
        &\phantom{=\,} - \frac{R_0^2}{A_0^2}
        \e^{(\sigma_B^2-\sigma_R^2)} %
        \e^{((\sigma_R c_{31} - \sigma_B)^2 + \sigma_R^2 c_{32}^2 +
          \sigma_R^2 c_{33}^2 + (\sigma_R c_{41} - \sigma_B c_{21})^2
          + (\sigma_R c_{42} - \sigma_B c_{22})^2 + \sigma_R^2
          c_{43}^2 + \sigma_R^2 c_{44}^2)/2}
        \\
        &= \rho_B \sigma_B^2 + 2\sigma_B \frac{R_0}{A_0}
        \e^{\sigma_B^2 - \rho_{RB} \sigma_B \sigma_R} (\rho_{RB} \sigma_R - \rho_B
        \sigma_B)% \\
        % & \phantom{=\,}
        + \frac{R_0^2}{A_0^2}
        \e^{2\sigma_B^2 - 2 \rho_{RB} \sigma_B \sigma_R} 
          \left(\e^{\rho_B \sigma_B^2 + \rho_R \sigma_R^2 -
              2\rho_{RB}\sigma_B\sigma_R} -1\right).
    \end{aligned}
  \end{multline*}
  The correlation is obtained by dividing by the respective variance.

  For the correlation of the modified bank's balance sheet with a
  single RE loan, we obtain
  \begin{multline*}
    \text{Cov}\left(\ln(A_1^i) + \frac{R_1^i}{A_1^i}, \ln(R_1)\right)
    \\%
    \begin{aligned}
      &=\text{Cov}\left(\sigma_B W_1^{i,B} + \frac{R_0}{A_0}
          \e^{(\sigma_B^2-\sigma_R^2)/2 + \sigma_R W_1^{i,R} -
            \sigma_B W_1^{i,B}}, \sigma_R W_1^{j,R}\right)\\
        &= \rho_{RB} \sigma_B\sigma_R + \frac{R_0}{A_0} \e^{(\sigma_B^2 -
          \sigma_R^2)/2} \sigma_R \text{Cov}\left(\e^{\sigma_R
            W_1^{i,R} - \sigma_B W_1^{i,B}}, W_1^{j,R}\right)\\
        &= \rho_{RB} \sigma_B\sigma_R + \sigma_R \frac{R_0}{A_0} \e^{(\sigma_B^2 -
          \sigma_R^2)/2} \text{Cov}\left( \e^{(\sigma_R
            c_{31}-\sigma_B) Z_1 + \sigma_R c_{32}Z_2 + \sigma_R c_{33} Z_3}, c_{41} Z_1
          + c_{42} Z_2 + c_{43} Z_3 + c_{44} Z_4\right)\\
        &= \rho_{RB} \sigma_B\sigma_R + \sigma_R \frac{R_0}{A_0} \e^{(\sigma_B^2 -
          \sigma_R^2)/2} \Big\{ (\sigma_R c_{31} - \sigma_B) c_{41}
          \e^{((\sigma_R c_{31} - \sigma_B)^2 + \sigma_R^2 c_{32}^2 +
            \sigma_R^2 c_{33}^2)/2} \Big.\\%
          &\phantom{=\,} \Big.+ \sigma_R c_{32} c_{42} \e^{(\sigma_R^2 c_{32}^2 +
            (\sigma_R c_{31} - \sigma_B)^2 + \sigma_R^2 c_{33}^2)/2} %
          + \sigma_R c_{33} c_{43} \e^{(\sigma_R^2 c_{33}^2 +
            (\sigma_R c_{31}-\sigma_B)^2 + \sigma_R^2 c_{32}^2)/2}
          \Big\}\\
           &= \rho_{RB} \sigma_B\sigma_R + \sigma_R \frac{R_0}{A_0} \e^{(\sigma_B^2 -
          \sigma_R^2)/2} \left( (\sigma_R c_{31} - \sigma_B) c_{41} +
          \sigma_R c_{32} c_{42} + \sigma_R c_{33} c_{43}\right)
          \e^{((\sigma_Rc_{31}-\sigma_B)^2 + \sigma_R^2 c_{32}^2 +
            \sigma_R^2 c_{33}^2)/2}\\
           &= \rho_{RB} \sigma_B\sigma_R + \sigma_R \frac{R_0}{A_0}
           \e^{\sigma_B^2 - \rho_{RB}\sigma_B \sigma_R} (\rho_R \sigma_R - \rho_{RB} \sigma_B). 
    \end{aligned}
  \end{multline*}
  The correlation is obtained by dividing by the respective standard
  deviations.
\end{proof}

\bibliographystyle{abbrvnamed}%
\bibliography{finance}

\begin{thebibliography}{}

\bibitem[\protect\citeauthoryear{Andersen and Sidenius}{2004}]{Andersen2004}
L.~Andersen and J.~Sidenius.
\newblock Extensions to the gaussian copula: Random recovery and random factor
  loadings.
\newblock {\em Journal of Credit Risk}, 1(1):29--70, 2004.

\bibitem[\protect\citeauthoryear{Bluhm \bgroup \em et al.\egroup
  }{2003}]{Bluhm2003}
C.~Bluhm, L.~Overbeck, and C.~Wagner.
\newblock {\em An Introduction to Credit Risk Modeling}.
\newblock Chapman \& Hall/CRC, London, 2003.

\bibitem[\protect\citeauthoryear{Bluhm}{2003}]{Bluhm2003a}
C.~Bluhm.
\newblock {CDO} modeling: Techniques, examples and applications.
\newblock Working Paper, December 2003.

\bibitem[\protect\citeauthoryear{Dorfleitner \bgroup \em et al.\egroup
  }{2011}]{Dorfleitner2011}
G.~Dorfleitner, M.~Leidl, and C.~Priberny.
\newblock Microcredit as an asset class: structured microfinance.
\newblock In {\em Mobilising Capital for Emerging Markets}, pages 137--154.
  Springer, 2011.

\bibitem[\protect\citeauthoryear{Duffie and Singleton}{2003}]{Duffie2003}
D.~Duffie and K.~J. Singleton.
\newblock {\em Credit Risk: Pricing, Measurement, and Management}.
\newblock Princeton University Press, 2003.

\bibitem[\protect\citeauthoryear{D{\"u}llmann \bgroup \em et al.\egroup
  }{2008}]{Duellmann2008}
K.~D{\"u}llmann, M.~Scheicher, and C.~Schmieder.
\newblock Asset correlations and credit portfolio risk: an empirical analysis.
\newblock {\em Journal of Credit Risk}, 4(2):37--63, 2008.

\bibitem[\protect\citeauthoryear{{Global Climate Partnership
  Fund}}{2019}]{GCPF2019}
{Global Climate Partnership Fund}.
\newblock Quarterly {R}eport {Q}3 2019.
\newblock Report, 2019.

\bibitem[\protect\citeauthoryear{Gordy}{2003}]{Gordy2003}
M.~B. Gordy.
\newblock A risk-factor model foundation for ratings-based bank capital rules.
\newblock {\em Journal of Financial Intermediation}, 12(3):199--232, 2003.

\bibitem[\protect\citeauthoryear{Greenberg \bgroup \em et al.\egroup
  }{2004}]{Greenberg2004}
A.~Greenberg, D.~O'Kane, and L.~Schloegl.
\newblock {LH+}: a fast analytical model for {CDO} hedging and risk management.
\newblock {\em Lehman Brothers Quantitative Credit Research Quarterly},
  Q2:19--31, 2004.

\bibitem[\protect\citeauthoryear{Gregory and Laurent}{2004}]{Gregory2004}
J.~Gregory and J.-P. Laurent.
\newblock In the core of correlation.
\newblock {\em RISK}, 17(10):87--91, 2004.

\bibitem[\protect\citeauthoryear{Hull and White}{2007}]{Hull2007}
J.~Hull and A.~White.
\newblock Dynamic models of portfolio credit risk: A simplified approach.
\newblock Working Paper, May 2007.

\bibitem[\protect\citeauthoryear{Hull and White}{2010}]{Hull2010}
J.~Hull and A.~White.
\newblock The risk of tranches created from mortgages.
\newblock {\em Financial Analysts Journal}, 66(5):54--67, 2010.

\bibitem[\protect\citeauthoryear{Krauss and Walter}{2009}]{Krauss2009}
N.~Krauss and I.~Walter.
\newblock Can microfinance reduce portfolio volatility?
\newblock {\em Economic Development and Cultural Change}, 58(1):85--110, 2009.

\bibitem[\protect\citeauthoryear{Li}{2000}]{Li2000}
D.~X. Li.
\newblock On {D}efault {C}orrelation: A {C}opula {F}unction {A}pproach.
\newblock {\em The Journal of Fixed Income}, pages 43--54, March 2000.

\bibitem[\protect\citeauthoryear{Merton}{1974}]{Merton1974}
R.~C. Merton.
\newblock On the pricing of corporate debt: The risk structure of interest
  rates.
\newblock {\em The Journal of Finance}, 29(2):449--470, May 1974.

\bibitem[\protect\citeauthoryear{O'Kane}{2008}]{OKane2008}
D.~O'Kane.
\newblock {\em Modelling Single-name and Multi-name Credit Derivatives}.
\newblock Wiley, 2008.

\bibitem[\protect\citeauthoryear{{S\&P Global Ratings}}{2015}]{SP2015}
{S\&P Global Ratings}.
\newblock Global methodologies and assumptions for corporate cash flow and
  synthetic {CDO}s.
\newblock Report, September 2015.

\bibitem[\protect\citeauthoryear{{S\&P Global Ratings}}{2019}]{SP2019}
{S\&P Global Ratings}.
\newblock 2018 annual global corporate default and rating transition study.
\newblock Report, April 2019.

\bibitem[\protect\citeauthoryear{Vasicek}{1987}]{Vasicek1987}
O.~Vasicek.
\newblock Probability of loss on loan portfolio.
\newblock KMV Corporation, 1987.

\bibitem[\protect\citeauthoryear{Vasicek}{1991}]{Vasicek1991}
O.~Vasicek.
\newblock Limiting loan loss probability distribution.
\newblock KMV Corporation, 1991.

\end{thebibliography}
\end{document}